\documentclass[twoside,12pt,reqno]{amsart}

\newcommand{\commentout}[1]{}

\DeclareSymbolFont{symbols}{OMS}{cmsy}{m}{n}

\usepackage[utf8]{inputenc}
\usepackage[T1]{fontenc}

\usepackage{mathtools}
\mathtoolsset{showonlyrefs}
\usepackage{mathrsfs}
\usepackage{mathtools}
\usepackage{amsmath,amssymb,amsthm}
\usepackage{amsfonts}
\usepackage{relsize}

\usepackage{enumerate}
\usepackage{graphicx,tikz}
\usetikzlibrary{arrows,scopes}
\usetikzlibrary{calc,cd}
\usetikzlibrary{decorations.markings}

\usepackage{setspace}

\usepackage{hyperref}
\usepackage[alphabetic,initials]{amsrefs}

\newtheorem{thm}{Theorem}[section]
\newtheorem{lem}[thm]{Lemma}
\newtheorem{cor}[thm]{Corollary}
\newtheorem{prop}[thm]{Proposition}
\theoremstyle{definition}

\newtheorem{example}[thm]{Example}
\theoremstyle{remark}

\newcommand{\cB}{\mathcal{B}}

\newcommand{\cC}{\mathcal{C}}

\newcommand{\A}{\mathcal{A}}

\newcommand{\RR}{\mathbb{R}}
\newcommand{\TT}{\mathbb{T}}

\newcommand{\ZZ}{\mathbb{Z}}
\newcommand{\NN}{\mathbb{N}}

\DeclareMathOperator{\Rep}{Rep}

\DeclareMathOperator{\Hom}{Hom}

\DeclareMathOperator{\Ad}{Ad}

\DeclareMathOperator{\id}{id}

\DeclareMathOperator{\sign}{sign}

\DeclareMathOperator{\Aut}{Aut}

\newcommand{\ima}{\mathrm{i}}

\usepackage{xspace}
\DeclareRobustCommand{\eg}{e.g.\@\xspace}
\DeclareRobustCommand{\cf}{cf.\@\xspace}
\DeclareRobustCommand{\ie}{i.e.\@\xspace}
\makeatletter
\DeclareRobustCommand{\etc}{%
    \@ifnextchar{.}%
        {etc}%
        {etc.\@\xspace}%
}
\makeatother

\def\u1net{{\A_\RR}}

\DeclareMathOperator{\ord}{ord}
\DeclareMathOperator{\Lex}{Lex}

\DeclareMathOperator{\Out}{Out}
\DeclareMathOperator{\Vect}{Hilb}
\DeclareMathOperator*{\Hilb}{Hilb}
\DeclareMathOperator{\Irr}{Irr}
\newcommand{\tRep}[1]{#1\text{--}\Rep}

\newcommand{\QQ}{\mathbb{Q}}

\newcommand{\cZ}{\mathcal{Z}}

\usepackage{soul}

\begin{document}
\date{\today}
\dateposted{\today}
\title[Anomalies of Holomorphic Cyclic Permutation Orbifolds]{A Remark About the Anomalies of Cyclic Holomorphic Permutation Orbifolds}

\address{Department of Mathematics, Morton Hall 321, 1 Ohio University, Athens, OH 45701, USA}
\author{Marcel Bischoff}
\email{bischoff@ohio.edu}
\email{marcel@localconformal.net}
\thanks{Supported in part by NSF DMs Grant 1700192/1821162}

\begin{abstract}
Using a result of Longo and Xu, we show that the anomaly arising from a cyclic permutation orbifold of order 3 of a holomorphic conformal net $\A$ with central charge $c=8k$  
depends on the ``gravitational anomaly'' $k\pmod 3$.
In particular, the  conjecture that holomorphic permutation orbifolds are non-anomalous and therefore a stronger conjecture of M\"uger about 
braided crossed $S_n$-categories
arising from permutation orbifolds of completely rational conformal nets are wrong.
More general, we show that cyclic permutations of order $n$ are non-anomalous if and only if 
$3\nmid n$ or $24|c$. 
We also show that all cyclic permutation gaugings of $\Rep(\A)$ arise from conformal nets.
\end{abstract}

\maketitle
\setcounter{tocdepth}{3}

\newcommand{\bs}{$\backslash$}
\newcommand{\CS}{/\!\!/}
\renewcommand{\L}{\mathrm{L}}

\newcommand{\tikzmath}[2][0.50]
{\vcenter{\hbox{\begin{tikzpicture}[scale=#1] #2\end{tikzpicture}}}
}
\newcommand{\colM}{black!20}

\newcommand{\colMa}{orange!50}
\newcommand{\colMab}{green!50}
\newcommand{\colMb}{red!30}
\renewcommand{\colMa}{black!20} \renewcommand{\colMab}{black!35} \renewcommand{\colMb}{black!30}
\newcommand{\colN}{black!10}
\newcommand{\mydot}[1]{\begin{scope}[shift={#1}] \fill[shift only] (0,0) circle (1.5pt); \end{scope}}

\newcommand{\cV}{\mathcal{V}}
\noindent

\newcommand{\cocyc}[1]{\omega_{#1}}
\section{Orbifolds and anomalies}
Conformal nets axiomatize chiral conformal field theory in the framwork of algebraic quantum field theory using von Neumann algebras. 
There is a notion of a completely rational conformal net \cite{KaLoMg2001}
whose representation category $\Rep(\A)$ is a modular tensor category.
Let $\A$ be a \textbf{holomorphic conformal net}, \ie 
a completely rational conformal net
with trivial representation category $\Rep(\A)\cong\Vect$.
Here we denote by $\Vect$ the trivial unitary modular tensor category of
finite-dimensional Hilbert spaces.
Let $G\leq \Aut(\A)$ be a finite group of automorphisms of the net $\A$,
see \cite{Xu2000-2,Mg2005}.
Then it is well-known \cite{KaLoMg2001,Mg2005,Mg2010,Bi2015,Bi2018} that there is a unique class $[\omega]\in H^3(G,\TT)$, 
such that the category of $G$-twisted representations of $\A$ denoted by 
$\tRep G^I(\A)$ is tensor equivalent to the category $\Vect_G^\omega$ of $G$-graded finite-dimensional Hilbert spaces
with associator given by $\omega$.
More precisely, for every $g\in G$ there is an irreducible $g$-twisted representation $\beta_g$
localized in $I$, which is unique up to conjugation by a unitary.
Then $g\mapsto [\beta_g]\in \Out(\A(I))$ is a \emph{$G$-kernel} and it
follows that $\beta_g\beta_h=\Ad u_{g,h} \beta_{gh}$ for unitaries 
$(u_{g,h})_{g,h\in G}$ and that $\omega\colon G\times G\times G \to 
\TT$ defined by
$\omega(g,h,k)\cdot 1=u_{g,h}u_{gh,k}u_{g,hk}^{-1}\beta_g(u_{h,k})^{-1}$
is a cocycle.
The class $[\omega]$ is called the \textbf{anomaly} of $G$ 
and we say that $G$ acts \textbf{non-anomalous} if $\omega$ is a coboundary.
Furthermore, the \emph{orbifold} or \emph{fixed point net} $\A^G$ has a 
representation category $\Rep(\A^G)$ which is braided equivalent 
to the Drinfel'd center $\mathcal Z(G,\omega)= Z(\Vect_G^\omega)$.

We denote by $S_n$ the symmetric group on $\{1,\ldots,n\}$.
Let $\A$ be a holomorphic net and  $G\leq S_n$, then $G$ acts by permutation on 
$\A^{\otimes n}$. 
It seems to be widely believed that this action should be non-anomalous. 
But Johnson-Freyd argued that this conjecture is false \cite{Jo2017}
and we give a counter-example in the framework of conformal nets where  
the permutation action picks up what can be thought of a \textbf{gravitational anomaly}\footnote{\cf \cite[Section 1.4]{Wi2007} for how this name might be justified, namely he asks that our $k$ equals $0\pmod 3$ in order for the chiral CFT to be dual to quantum gravity.%
} $k\equiv c/8\pmod 3\in \ZZ_3$, where $c$ is the central charge of $\A$.

This note is an extension of an unpublished note (consisting essentially of 
Section \ref{sec:Unpublished}) circulated in 2017.
The results were announced April 15th, 2018 at the AMS Sectional Meeting
at Vanderbilt University, Nashville, TN.
Shortly after that, a more general result appeared in a preprint
by Evans and Gannon \cite{EvGa2018}.
\subsection*{Acknowledgements}
The original note is based on communication with 
Theo Johnson-Freyd
who told me that permutation orbifolds can be anomalous and gave me the 
counterexample arising from the $E_8$ lattice \cite{Jo2017}. 
I am thankful for the communication and explaining me his work.

\section{Cyclic permutations of order 3}
\label{sec:Unpublished}
\subsection{Twisted doubles of $\ZZ_3$.}
Recall that a unitary fusion category is called pointed if all simple objects are invertible.
Pointed unitary fusion categories with $\ZZ_3$-fusion rules are classified by 
$H^3(\ZZ_3,\TT)=\{[\cocyc i]:i\in \ZZ/3\ZZ\} \cong \ZZ_3$.
Their Drinfel'd centers $\mathcal{Z}(\ZZ_3,\cocyc i):= Z(\Vect^{\cocyc i}_G)$ are pointed. 
Indeed, it can be easily checked that they are braided equivalent to the 
pointed unitary modular tensor categories
$\cC(G_i,q_i)$, respectively, where $(G_i,q_i)$ are the metric groups
given in Table \ref{tab:Z3} and $\cC(G,q)$ is 
the braided fusion category associated to $(G,q)$, see
Appendix \ref{app}, in particular Proposition \ref{prop:PointedUMTC}.
\begin{table}[h!] 
\begin{tabular}{c|c|c}
  $i$ & $G_i$ & $q_i\colon G_i\to\QQ/\ZZ$\\
  \hline
  0 & $\ZZ_3\times \ZZ_3$ & $q_0(x,y) \equiv  xy/3\pmod 1$\\

  1 & $\ZZ_9$ & $q_1(x)\equiv4x^2/9\pmod 1$\\
  2 & $\ZZ_9$ & $q_2(x)\equiv8x^2/9\pmod 1$\\
\end{tabular}
\medskip
\caption{Twisted doubles of $\ZZ_3$}
  \label{tab:Z3}
\end{table}

\subsection{The anomaly}
By a conformal net we mean a diffeomorphism covariant net 
on the circle,  see \eg \cite{KaLo2006}.
Let $\A$ be a holomorphic conformal net.
Since $\A$ is diffeomorphism covariant we can assign a central charge $c>0$.
It is conjectured that if $\A$ is holomorphic, then $c\equiv 0\pmod 8$ and it is a theorem that $c\equiv 0\pmod 4$ \cite{KaLoXu2005}.
We will from now on assume that the central charge $c$ of $\A$ fulfills  $c\in 8\NN$.
If $\A$ is holomorphic, then any tensor power $\A^{\otimes n}$ is holomorphic 
\cite{KaLoMg2001}.
Let $\sigma\in S_n$ be a permutation. 
Then there is an element $\sigma\in\Aut(\A^{\otimes n})$ given by
\begin{align}
  x_1\otimes x_2\otimes \cdots \otimes x_n \mapsto 
  x_{\sigma(1)}\otimes x_{\sigma(2)}\otimes \cdots \otimes x_{\sigma(n)}
\end{align}
see \eg \cite{LoXu2004}.

\begin{prop} 
  Let $\A$ be a diffeomorphism covariant holomorphic net with $c=8k$, 
  and let $\ZZ_3\cong \langle \tau \rangle \leq \Aut(\A^{\otimes 3})$ be the group generated by the cyclic permutation $\tau=(123)$.
  Then the anomaly of $\langle \tau\rangle$ is $\cocyc {2k}$, %
  \ie
  $\tRep{\langle\tau\rangle}(\A^{\otimes 3})$ is tensor equivalent to  $\Vect_{\ZZ_3}^{\cocyc {2k}}$%
  and $\Rep((\A^{\otimes 3})^{\langle \tau\rangle})$ is braided equivalent to $
  \mathcal Z(\ZZ_3,\cocyc {2k})$.%
\end{prop}
\begin{proof}
  It is enough to show that $\cC:=\Rep((\A^{\otimes 3})^{\langle \tau\rangle})$ is
  braided equivalent to $\cC(G_{2k},q_{2k})$. 
  But this follows from \cite[Theorem 6.3e]{LoXu2004} which gives that the 
  spins in
  $\Rep((\A^{\otimes 3})^{\langle \tau\rangle})$ coming from twisted sectors 
  $\alpha_i$ are $h_i=i/3+8k/9$ for $i=0,1,2$ and then 
  $q([\alpha_i])\equiv h_i\pmod 1$ by the spin--statistic theorem 
  \cite{GuLo1996}.
  This readily identifies $\cC$ to be braided equivalent with 
  $\cC(G_{2k},q_{2k})$.
\end{proof}
\begin{example}
  Let $\A_{E_8}$ be the conformal net associated with the even lattice $E_8$
  \cite{DoXu2006}.
  Then $\tRep{\langle\tau\rangle}(\A_{E_8}^{\otimes 3})$
  is tensor equivalent to $\Vect^{\omega_2}_{\ZZ_3}$ with $[\omega_2]$ a generator of $H^3(\ZZ_3,\TT)$.
  Thus $\Rep((\A_{E_8}^{\otimes 3})^\tau)$ is braided equivalent to 
  $\cZ(\ZZ_3,\omega_2)$.
\end{example}

\begin{example}
  Let $\A$ be a holomorphic net with central charge $c=8k$.
Let $S_3\leq \Aut(\A^{\otimes 3})$ be the group of all permutations.
Since $H^3(S_3,\TT))\cong H^3(\ZZ_3,\TT)\oplus H^3(\ZZ_2,\TT)$, where the 
isomorphism comes from restriction, it follows that $S_3$ is anomalous unless
$k=0\pmod 3$.
In particular,  $\Rep((\A_{E_8}^{\otimes 3})^{S_3})$ is braided equivalent to $\cZ(S_3,\tilde \omega)$ 
  for some $[\tilde\omega]\in H^3(S_3,\TT)$ of order 3.
\end{example}

In particular, the conjecture by M\"uger \cite[Appendix 5, Conjecture 6.3]{Tu2010} that states that for every completely rational conformal net $\A$ 
the category of $S_n$-twisted representations
$\tRep{S_n}(\A^{\otimes n })$ up to tensor equivalence depends only on the modular tensor category $\Rep(\A)$ is wrong.

\section{Cyclic holomorphic orbifolds}
The argument can be generalized to arbitrary cyclic extensions and we get
the following result.
\begin{prop}
  \label{prop:Cyclic}
  Let $\A$ be a holomorphic net with central charge $c=8k$ for some $k\in\NN$. 
  Let $\alpha$ be a cyclic permutation of order $n$ on 
  $\A^{\otimes m}$.
  Then the action of 
  $\langle\alpha\rangle\cong \ZZ_n$ on $\A^{\otimes m}$ is non-anomalous if and only if 
  $3\nmid n$ or $24 \mid c$.
\end{prop}

\subsection{Cyclic homolorphic twisted orbifolds}
We have the following application of 
  Proposition \ref{prop:Cyclic}.

If $\A$ is holomorphic and $G\leq \Aut(\A)$ 
non-anomalous we can form the so-called \textbf{twisted orbifold}
$\A\CS G$ as described in \cite{Bi2018}
by lifting the $G$-kernel given by $\tRep G(\A)$
to a homorphism $G\hookrightarrow\Aut(\A(I))$
or in other words by choosing a trivilization.

In our concrete case, this can be easier described. 
Namely,
$\Rep((\A^{\otimes m})^{\langle\alpha\rangle})$ is braided equivalent to 
$\cC(\ZZ_n\times \hat \ZZ_n,q_\mathrm{st})$ with the quadratic form 
$q_\mathrm{st}(g,\chi)=\chi(g)$,
such that the Lagrangian subgroup 
$\ZZ_n\times \{\chi_0\}$ gives $\A^{\otimes m}$.
We have a second Lagrangian subgroup $\{0\}\times \hat \ZZ_n$ 
which gives a new holomorphic net $\A^{\otimes n}\CS \langle\alpha\rangle$ which 
is the \textbf{twisted orbifold net} $\A^{\otimes m}\CS \langle\alpha\rangle$
of $\A^{\otimes m}$ 
with respect to $\langle\alpha\rangle$.
Thus we have:
\begin{prop}
  Let $\A$ be a holomorphic net with central charge $c\in 8\NN$. 
  Let $\alpha$ be a cyclic permutation of order $n$ on 
  $\A^{\otimes m}$.
  If 
  $3\nmid n$ or $24 \mid c$, we have a holomorphic net
  given by the twisted orbifold 
  $\A^{\otimes m}\CS \langle\alpha\rangle$.
\end{prop}
\begin{example}
  $(\A_{E_8}\otimes \A_{E_8})/\!\!/\langle\tau_2\rangle$
is isomorphic to $\A_{D^+_{16}}$.
\end{example}

\subsection{Determining the anomalies}
We now proceed to prove Proposition \ref{prop:Cyclic}.
Let $\A$ be a holomorphic net and let $\tau_n$ be the cyclic permutation 
\begin{align}
  x_1\otimes x_2\otimes \cdots \otimes x_n \mapsto x_2\otimes x_3\otimes \cdots \otimes x_1\,.
\end{align}
Then  $\tau_n$ yields an inner symmetry
$\tau_n\in\Aut(\A^{\otimes n})$, see \eg \cite{LoXu2004}.

For $G\leq \Aut(\A)$ and $g\in G$ 
we denote by $\Rep(\A^G)_g$ the category of representations 
coming from restrictions of $\tRep g(\A)$.

\begin{lem}
  Let $\A$ be a holomorphic net with central charge $c=8k$ for some 
  $k\in\NN$.
  \begin{enumerate}
    \item 
      For $3\nmid n$  the spectrum of $h_\alpha$ with $\alpha\in
  \Rep((\A^{\otimes n})^{\langle\tau_n\rangle})_{\tau_n}$  
  is $\{0,\frac1n,\ldots,\frac{n-1}n\}\pmod 1$.
    \item 
      For $n=3m$ the spectrum of $h_\alpha$ with 
  $\alpha\in
  \Rep((\A^{\otimes n})^{\langle\tau_n\rangle})_{\tau_n}$  
  is $\{0,\frac1n,\ldots,\frac{n-1}n\}\pmod 1$ if and only if $c\equiv0\pmod {24}$.
    \item
  For $k= \pm1\equiv \frac c8\pmod 3$ there is an 
  $\alpha\in \Rep((\A^{\otimes n})^{\langle\tau_n\rangle})_{\tau_n}$  
  with $h_\alpha\equiv - \frac k{3n}\pmod 1$.
  \end{enumerate}
\end{lem}
\begin{proof}
  Let $n=3\ell\pm 1$. 
  Then using \cite[Theorem 6.3e]{LoXu2004} we have
  \begin{align} 
    h_i&=\frac{i}n+\frac{n^2-1}{24n}c  
    =\frac{i+k\ell(3\ell\pm 2)}n
  \end{align}
  thus we have $(1)$.
  Now let $n=3\ell$, then 
  \begin{align}
    h_i&=\frac{i}n+\frac{n^2-1}{24n}c  
    =\frac{i}n+\frac{9\ell^2-1}{9\ell}k
    \\&=\begin{cases}
      \frac{i-m}n \pmod 1
 &k=3m \\
  -\frac{3(m-i)\pm 1}{3n}\pmod 1 
 &k=3m\pm 1\,.
    \end{cases}
  \end{align}
  Thus we have (2) and 
  since $3\nmid 3(m-i)\pm 1$ we get (3).
\end{proof}

\begin{prop}
  Let $\A$ be a holomorphic net.
  If $3\nmid n$ then $\tRep{\langle\tau_n\rangle}(\A)$ is tensor equivalent to $\Vect_{\ZZ_n}$.
\end{prop}
\begin{proof}
  Since there is a $\tau_n$-twisted representation  $\beta$  with $h_\beta=0\pmod 1$ from Lemma \ref{lem:CyclicExtension} it follows that 
  $\Rep((\A^\otimes n )^{\langle\tau_n\rangle})$ is braided equivalent to
  $\cC(\ZZ_n\times\hat\ZZ_n,q)$
  and because the Lagrangian subgroup lives in the zero graded part we have 
   $\tRep{\langle\tau_n\rangle}(\A)$ is tensor equivalent to $\Vect_{\ZZ_n}$
   again by Lemma \ref{lem:CyclicExtension}.
\end{proof}
\begin{prop}
  \label{prop:Cocycles}
  Let $\A$ be a holomorphic net of central charge $c=8k$
  and $n=3m$ for some $m,k\in\NN$.
  \begin{enumerate}
    \item 
$\Rep((\A^{\otimes n})^{\langle\tau_n\rangle})$ is braided equivalent 
  to $\cC(\ZZ_{9m}\otimes \ZZ_m,q_\mp)$ 
  with 
  \begin{align}
    q_\pm(x,y) = \frac{\pm x^2}{9m}\mp\frac{y^2}{m}
  \end{align}
  for $k\equiv \pm 1\pmod 3$.
\item    $\Rep((\A^{\otimes n})^{\langle\tau_n\rangle})$ is braided equivalent 
  to $\cC(\ZZ_{3m}\otimes \hat \ZZ_{3m},q)$ for $k\equiv 0\pmod  3$.
    \item  $\tRep{\langle\tau_n\rangle}(\A)$ is tensor equivalent to
      $\Vect_{\ZZ_n}^{\omega_k}$ with 
      $[\omega_k]=-km[\omega_0]$ for a generator $[\omega_0]$
      of $H^3(\ZZ_n,\QQ/\ZZ)\cong \ZZ_n$. 
  \end{enumerate}
\end{prop}
In Figure \ref{fig:1}, we demonstrate the twisted fusion rules 
depending on $k$ in an example.

\begin{proof} 
  (2) is proved as before.
  We note that the cocycle has order three, since $\tau_{3\ell}$
  equals $\tau_\ell$ on $(\A^{\otimes 3})^{\otimes\ell}$ since $\A^{\otimes 3}$ has central 
  central charge $c=24k$. 
  So there are only two choices for the cocycle which are distinguished by 
  the values of $h$, see Appendix \ref{app}, which proves (1) and (3).
\end{proof}
\begin{cor}
  Let $\A$ be a holomorphic net. The action of 
  $\langle\tau_n\rangle\cong \ZZ_n$ on $\A^{\otimes n}$ is non-anomalous if and only if 
  $3\nmid n$ or $24 \mid c$.
\end{cor}
In particular, we have proven Proposition \ref{prop:Cyclic},
since any cyclic permutation in $\A^{\otimes m}$ is conjugate 
to $\tau_n\otimes \id$.

\section{All gaugings for cyclic permutation orbifolds}
Given a unitary modular tensor category $\cC$ 
we can consider $\cC^{\boxtimes n}$ which has a categorical action
of any subgroup $G\leq S_n$. 
Recently,  T.~Gannon and C.~Jones have showed \cite{GaJo2018} that certain
obstructions vanish and that therefore
such a symmetry can always be gauged, \ie there is a with the categorical action compatible $G$-crossed braided extension $\cC\wr G\supset \cC^{\boxtimes n}$. 
The equivariantization 
$(\cC\wr G)^G$  is a new unitary modular tensor category, which correspond to gauging.
If $\cC=\Rep(\A)$ for a rational conformal net, then 
$\tRep G(A^{\otimes n})$ (where $G$ acts by permutations) 
is a  $G$-crossed braided extension 
and $\Rep((A^{\otimes n})^G)$ is a  special gauging.

Using cyclic orbifolds of rational (not necessarily holomorphic)
nets, we show that if a unitary modular tensor category 
$\cC$ is realized by conformal nets, 
then all $\ZZ_n$-permutation gaugings of $\cC$ are realized. 
\begin{prop}
  Consider the unitary modular tensor category 
  $\cC=\Rep(\A)$ for a rational conformal net $\A$.
  
  Then any
   unitary $\ZZ_n$-crossed braided extension   $\cC\wr\ZZ_n$ of $\cC^{\boxtimes n}$ 
where $\ZZ_n$ acts 
by cyclic permutations on $\cC^{\boxtimes n}$
   is realized as $\tRep{\ZZ_n}(\cB)$ 
  for some conformal net $\cB$ and  $\ZZ_n\hookrightarrow\Aut(\cB)$.

  In particular, any gauging of the cyclic permutation on $\cC^{\boxtimes n}$
are realized by a conformal net $\cB^{\ZZ_n}$.
\end{prop}
\begin{proof} 
  There are $n$ distinguished extensions $\cC\wr \ZZ_n$ \cite[Lemma 2.3]{EdJoPl2018}.
  One is realized by the cyclic permutation orbifold \cite{LoXu2004,KaLoXu2005} 
  $\ZZ_n\hookrightarrow \langle\tau_n\rangle\leq  \Aut(\A^{\otimes})$.
  Let $[\omega]\in H^3(\ZZ_n,\TT)\cong \ZZ_n$.
  Since 
  $\mathcal Z(\Hilb_{\ZZ_n}^\omega)$ is pointed, 
  by \cite[Theorem 3.6]{Bi2018} there is a conformal net associated with 
  a lattice $\A_L$ realizing $\mathcal Z(\ZZ_n,\omega)$.
  Then there is a $\hat \ZZ_n$-simple current extension $\cB_\omega\supset\A_L$
  and $\ZZ_n\hookrightarrow\Aut(\cB_\omega)$, such that 
  $\tRep{\ZZ_n}(\cB_\omega)\cong \Hilb_{\ZZ_n}^\omega$.

  Finally, $\tRep{\Delta(\ZZ_n)}(\A^{\otimes n}\otimes \cB_\omega)$ 
with $\Delta(\ZZ_n)\subset \ZZ_n\times\ZZ_n$ the diagonal subgroup gives all 
 $\ZZ_n$-crossed braided extensions  by varying the class $[\omega]$ 
using \cite[Proposition 3.4]{Bi2018}.
\end{proof} 
We note that the reconstruction program asks if for any unitary modular tensor category $\cC$ there is a conformal net realizing it.
In this perspective, the $H^3(\ZZ_n,\TT)$ freedom in gauging of cyclic 
permutations does not give any obstructions.

\newcommand{\midarrow}{\tikz \draw[->] (0,0) -- +(.1,0);}
\newcommand{\midarroww}{\tikz \draw[->>] (0,0) -- +(.1,0);}
\begin{figure}
  \begin{align*}
&\begin{tikzpicture}[scale=.7]
  \draw[black](-.5,-.5)-- node[sloped] {\midarrow} (-.5,8.5);
  \draw[black](8.5,-.5)--node[sloped] {\midarrow} (8.5,8.5);
  \draw[black](-.5,8.5)--node[sloped] {\midarroww} (8.5,8.5);
  \draw[black](-.5,-.5)-- node[sloped] {\midarroww}(8.5,-.5);
  \path [clip] (-.5,-.5)--(-.5,8.5)--(8.5,8.5)--(8.5,-.5)
  --(-.5,-.5);
  \fill[black!3] (-.5,-.5)--(-.5,8.5)--(8.5,8.5)--(8.5,-.5)
  --(-.5,-.5);
  \foreach \x in {-2,-1,...,34}
  {
    \draw[black!50] (\x,-.5)--(\x,8.5);
  }
  \foreach \y in {0,1,...,9}
  {
    \draw (-.5,\y)--(8.5,\y);
  }
  \foreach \x in {0,1,...,8}
  {
    \foreach \y in {0,1,...,8}
    {
      \fill (\x,\y) circle (.1);
    }
  }
  \draw[black](-.5,-.5)-- node[sloped] {\midarrow} (-.5,8.5);
  \draw[black](8.5,-.5)--node[sloped] {\midarrow} (8.5,8.5);
  \draw[black](-.5,8.5)--node[sloped] {\midarroww} (8.5,8.5);
  \draw[black](-.5,-.5)-- node[sloped] {\midarroww}(8.5,-.5);
\end{tikzpicture}
\\&\qquad
\\&
\begin{tikzpicture}[scale=.7]
  \draw[black](-.5,-.6666)-- node[sloped] {\midarrow} (-.5,8.3333);
  \draw[black](8.5,2.3333)--node[sloped] {\midarrow} (8.5,11.33333);
  \draw[black](-.5,8.3333)--node[sloped] {\midarroww} (8.5,11.33333);
  \draw[black](-.5,-.6666)-- node[sloped] {\midarroww}(8.5,2.3333);
  \path [clip] (-.5,-.6666)--(-.5,8.3333)--(8.5,11.33333)--(8.5,2.3333)
  --(-.5,-.6666);
  \fill[black!3] (-.5,-.66666)--(-.5,8.3333)--(8.5,11.33333)--(8.5,2.3333) --(-.5,-.66666);
  \foreach \y in {-2,-1,...,34}
  {
    \draw[black!50] (-3,0.33333*\y-1)--(9,0.33333*\y+3);
  }
  \foreach \y in {0,3,...,9}
  {
    \draw[red,thick] (-.5,\y+2)--(8.5,\y+2);
  }
  \foreach \y in {0,3,...,12}
  {
    \draw[black!60!green,thick] (-.5,1+\y)--(8.5,1+\y);
  }
  \foreach \y in {0,3,...,9}
  {
    \draw[black!20!blue,thick] (-.5,\y)--(8.5,\y);
  }
  \foreach \x in {0,1,...,8}
  {
    \foreach \y in {0,1,...,11}
    {
      \fill (\x,\y) circle (.1);
    }
  }
  \draw[black](-.5,-.6666)-- node[sloped] {\midarrow} (-.5,8.3333);
  \draw[black](8.5,2.3333)--node[sloped] {\midarrow} (8.5,11.33333);
  \draw[black](-.5,8.3333)--node[sloped] {\midarroww} (8.5,11.33333);
  \draw[black](-.5,-.6666)-- node[sloped] {\midarroww}(8.5,2.3333);
\end{tikzpicture}
\end{align*}
\caption{$\ZZ_{9}\oplus\ZZ_9$ and $\ZZ_{27}\oplus\ZZ_9$ fusion rules for 
$\Rep((\A^{\otimes 9})^{\langle\tau_9\rangle})$
for central charge $c\equiv 0\pmod{24}$ and $c\equiv 8\pmod {24}$, respectively. The gravitational anomaly 
$\exp{\frac{2\ima\pi k}3}$ with $k=0,1$, respectively,
twists the torus.}
\label{fig:1}
\end{figure}
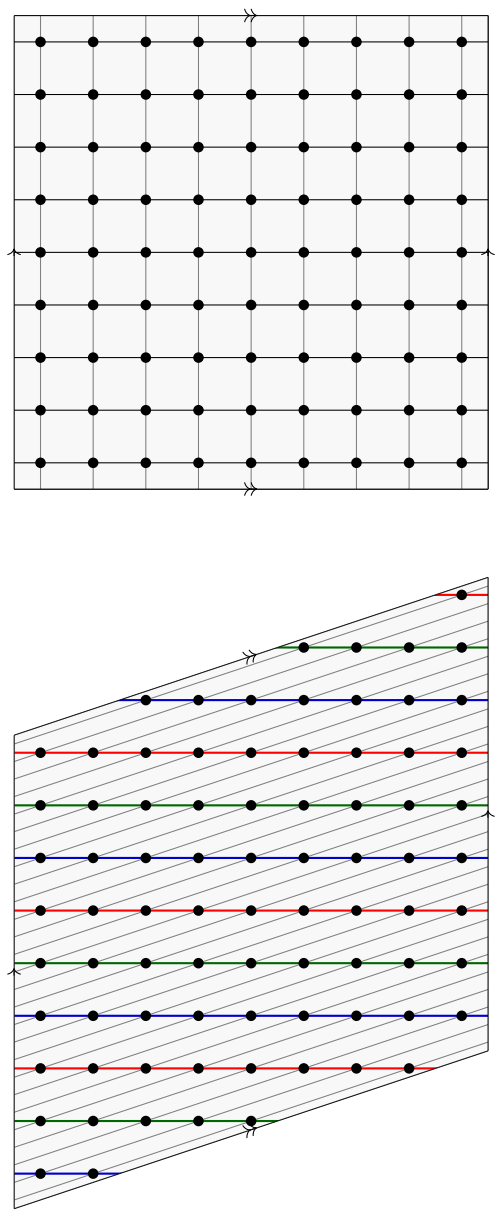

\commentout{
$$
\begin{tikzpicture}[scale=.5]
  \path [clip] (-.5,-.5)--(-.5,4.5)--(4.5,6.1666)--(4.5,1.166)
  --(-.5,-.5);
  \fill[black!10] (-.5,-.5)--(-.5,4.5)--(4.5,6.1666)--(4.5,1.166)
  --(-.5,-.5);
  \foreach \x in {0,1,...,4}
  {
    \foreach \y in {0,1,...,7}
    {
      \fill (\x,\y) circle (.1);
    }
  }

  \draw[dashed,red] (-.5,5.666)--(4.5,5.666);
  \draw[dashed,red,opacity=0.4]
  (-.5,4)--(4.5,4)
  (-.5,2.333)--(4.5,2.333)
  (-.5,0.666)--(4.5,0.666);
  \draw[dash dot,black!60!green] (-.5,3)--(4.5,3);
  \draw[dash dot,black!60!green,opacity=0.4]
  (-.5,4.666)--(4.5,4.666)
  (-.5,1.333)--(4.5,1.333)
  (-.5,-0.666)--(4.5,-0.666);
  \draw[black!20!blue] (-.5,2)--(4.5,2);
  \draw[black!20!blue,opacity=0.4] 
    (-.5,5.333)--(4.5,5.333)
    (-.5,3.666)--(4.5,3.666)
    (-.5,0.333)--(4.5,0.333);
  \draw[black!20!blue] 
    (-.5,6)--(4.5,6)
    (-.5,1)--(4.5,1);
  \draw[black!20!blue,opacity=0.4] 
    (-.5,4.333)--(4.5,4.333)
    (-.5,2.666)--(4.5,2.666)
    (-.5,-0.666)--(4.5,-.666);
  \draw[dashed,black!20!blue!red] 
    (-.5,0)--(4.5,0)
    (-.5,5)--(4.5,5);
  \draw[dashed,black!20!blue!red,opacity=0.4] 
    (-.5,3.333)--(4.5,3.333)
    (-.5,1.666)--(4.5,1.666)
\end{tikzpicture}
$$
}

\begin{appendix}
\section{Lagrangian extensions}
\label{app}
A \textbf{premetric group} $(A,q)$ consists of a finite abelian group $A$
which we see as an additive group
and a quadratic form $q\colon A\to \QQ/\ZZ$, \ie  $q(na)=n^2q(a)$
for all $a\in A$ and $n\in\ZZ$ 
and $\partial q(a,b) = q(a+b)-q(a)-q(b)$ is a bicharacter. 
A \textbf{metric group} is a premetric group $(A,q)$ with $\partial q$ non-degenerate.
A morphism $\tau\colon (A_1,q_1)\to (A_2,q_2)$ is a homorphism $\tau\colon A_1\to A_2$ 
with $ q_1=q_2\circ\tau$.

The following is well-known, see eg.\ \cite{JoSt1993} and \cite{EtGeNiOs2015}.
\begin{prop}
  \label{prop:PointedUMTC}
  Given a metric group $(A,q)$ there is an up to
  braided equivalence unique unitary modular tensor category denoted by
  $\cC(A,q)$ such that the braiding 
  $c_{X_a,X_a}=c_q(g)\cdot 1_{X_a\otimes X_a}$ and thus the 
  twist $\theta_{X_a} = \exp(2\pi \ima q(a))$ for all $a\in A$. 

  Conversely, given a pointed unitary modular tensor category $\cC$, 
  the finite set $A=\Irr(\cC)$ is an abelian group under the tensor product 
  and the braiding $c$ defines a quadratic form $q(g)\cdot 1_{X_g\otimes X_g}=
  c_{X_g,X_g}$ for every $g=[X_g]\in G$.
  Then $\cC$ is braided equivalent to $\cC(A,q)$.
\end{prop}

  We define 
  $H^3(A,\TT)_{\mathrm{ab}}=\ker(\psi^\ast)$,  
  where $\psi^\ast\colon H^3(G,\TT)\to \Hom(\Lambda^3
  G,\TT)$ is  given by
  \begin{align}
    \left[\psi^\ast([\omega])\right](x,y,z)&=
      \prod_{\omega\in\mathbb{S}_3}\omega(\sigma(x),\sigma(y),\sigma(z))
        ^{\sign(\sigma)}\,.
  \end{align}
  The Drinfel'd center $\cZ(A,\omega)= Z(\Vect^{\omega}_A)$ is pointed if and only if
  $[\omega]\in H^3(A,\TT)_\mathrm{ab}$
  \cite[Corollary 3.6]{NgMa2001}, see also \cite[Proposition 4.1]{Ng2003}.

Let $\hat B$ be an abelian group. 
A \textbf{Lagrangian extension}  of $\hat B$ is a triple $(A,q,\iota)$
consisting of a metric group $(A,q)$ with  $|A|=|\hat B|^2$ and 
a monomorphism $\iota\colon(\hat B,0)\hookrightarrow (A,q)$ of premetric groups.
The isomorphism classes of Lagrangian extensions of $\hat B$ form an abelian group 
$\Lex(\hat B)$ via the multiplication $(A_1,q_1,\iota_1)\boxplus(A_2,q_2,\iota_2)$, see \cite{DaSi2017-2} for details.
Given a Lagrangian extension $(A,q,\iota)$ of $\hat B$ we obtain a Lagrangian algebra
$L=\iota(\hat B)$ in $\cC(A,q)$ and $\cC(A,q)_L=\cC(A,q)_B$ is naturally 
isomorphic to $\Hilb^\omega_B$ for some $[\omega]\in H^3(B,\TT)_\mathrm{ab}$ and the map $(A,q,\iota)\to [\omega]$
gives an isomorphism $\Lex(\hat B)\to H^3(B,\TT)_{\mathrm{ab}}$ of abelian groups.
\begin{example} Let $A$ be an abelian group and $\hat A=\Hom(A,\QQ/\ZZ)$ 
  the dual group. 
  Then $(A\times \hat A, q_\mathrm{st},\iota)$ is an Lagrangian extension of 
  $A$, where $q_\mathrm{st}(a,\chi)=\chi(a)$ and $\iota\colon A\to A\otimes \hat A$ is the canonical inclusion.
  Note that the isomorphism class of $(A\times \hat A, q_\mathrm{st},\iota)$ is the unit under $\boxplus$ and thus correspond to the trivial cohomology class
  in $H^3(A,\TT)$.
\end{example}
\begin{lem}
  \label{lem:CyclicExtension}
  Let $(G,q,\iota)$ be a Lagrangian extension of $\ZZ_n$ 
  and consider the map $p\colon G\to G/\ZZ_n\cong \ZZ_n$.
  If there is a $x\in G$ with $p(x)$ a generator and $q(x)=0$,
  then $(G,q)\cong (\ZZ_n\times \hat \ZZ_n,q_\mathrm{st})$
  and $\cC(G,q)_{\iota(\ZZ_n)}$ is tensor equivalent to 
  $\Hilb_{\ZZ_n}$.
\end{lem}
\begin{proof}
  We claim that the order $\ord(x)$ of $x$ is $n$. 
  One the one hand, it is a multiple of $n$.
  On the other hand, $q(mx)\equiv 0\pmod 1$
  and thus $L'=\langle x\rangle$ is a isotropic subspace of $(G,q)$
  and thus $\ord(x) \leq n$.
  Then $\chi(n)=q(x+\iota(n))=q(x+\iota(n))-q(x)-q(\iota(n))=\partial(\iota(n),x)$ defines a character $\chi\colon \ZZ_n\to \QQ/\ZZ$ 
  and $\langle \chi\rangle =\hat \ZZ_n$ because $q$ is non-degenerate. 
  Finally,    
  $\iota(m)+nx\mapsto (m,n\chi)$
  gives an isomorphism of metric groups 
  $(G,q)\to (\ZZ_n\times\hat \ZZ_n,q_\mathrm{st})$.
\end{proof}

\begin{example}
  \label{ex:OrderThree}
Let $n=3m$ and consider the following Lagrangian extension of $\hat\ZZ_{3m}=\{\chi_j\}$
with $\chi_j\colon \ZZ_{3m}\to \QQ/\ZZ$ with $\chi_j(x)=\frac {jx}{3m}$.
We define Lagrangian extensions $(A_\pm,q_\pm,\iota_\pm)$ and
$(A_0,q_0,\iota_0)$
\begin{align}
  A_\pm&=\ZZ_{9m}\oplus\ZZ_m &
  q_\pm(x,y)&=\pm\frac{x^2}{9m}\mp\frac{y^2}{m}\\
  A_0&=\hat \ZZ_{3m}\oplus\ZZ_{3m}&
  q_0(\chi,x)&=\chi(x)
\end{align}
with $\iota_0$ the canonical embedding
$\iota_0\colon\hat{\ZZ_3}\to\hat\ZZ_3\oplus\ZZ_3=\A_0$.
Then with $j=(9m-3,1)$ is a simple current of order $3m$ and 
we have the short exact sequence
\begin{align}
  \{0\}\longrightarrow \hat \ZZ_{3m}\xrightarrow{\iota_\pm\colon \chi_1\mapsto j}
  A_\pm=\ZZ_{9m}\oplus\ZZ_m\longrightarrow (\ZZ_{9m} \oplus\ZZ_m)/\langle j\rangle\longrightarrow \{0\}\,.
\end{align}
We have the relations $(A_\pm,q_\pm,\iota_\pm)\boxplus (A_\pm,q_\pm,\iota_\pm)=(A_\mp,q_\mp,\iota_\mp)$ and 
$(A_+,q_+,\iota_+)\boxplus(A_-,q_-,\iota_-)=(A_0,q_0,\iota_0)$ which gives a subgroup of $\Lex(\hat \ZZ_{3m})$ isomorphic to $\ZZ_3$.
Let $[\omega_\pm]$ be the cohomology class associated with $(A_\pm,q_\pm)$, then
$[\omega_\pm]=\pm m[\omega]\in\langle[\omega]\rangle=H^3(\ZZ_{3m},\TT)$.
This are the cocycle arising in Proposition \ref{prop:Cocycles}.
\end{example}
\end{appendix}

\def\cprime{$'$}\newcommand{\noopsort}[1]{}

\address
\end{document}